\documentclass[conference]{IEEEtran}
\usepackage{amsthm}  
\usepackage{amsmath}
\usepackage{cite}
\usepackage{graphicx}
\usepackage{float}
\usepackage{times}
\usepackage{latexsym}
\usepackage{bm}
\usepackage{amssymb}
\usepackage{stfloats}
\usepackage{array}
\usepackage{fancyhdr}
\usepackage{cite,graphicx,amssymb}
\usepackage{citesort}
\usepackage{psfrag}
\usepackage{multirow}

\usepackage{color}
\usepackage{stfloats}
\usepackage{bm}  
\usepackage{upgreek}  

\newtheorem{theorem}{Theorem}

\newtheorem{lemma}{Lemma}

\newtheorem{corollary}{Corollary}

\newtheorem{proposition}{Proposition}

\newtheorem{remark}{Remark}

\ifCLASSINFOpdf
\else
\fi
\newcommand{\mC}{\mathcal {C}}
\newcommand{\mN}{\mathcal {N}}

\newcommand{\tE}{{\tt E}}
\newcommand{\tP}{{\tt P}}

\begin{document}
%
\title{Analysis of Outage Probabilities for Cooperative NOMA Users with Imperfect CSI}




%
\author{\IEEEauthorblockN{Xuesong~Liang\IEEEauthorrefmark{1}\IEEEauthorrefmark{2},
Xinbao~Gong\IEEEauthorrefmark{3},
Yongpeng~Wu\IEEEauthorrefmark{3},
Derrick~Wing~Kwan~Ng\IEEEauthorrefmark{4} and
Tao~Hong\IEEEauthorrefmark{5}}
\IEEEauthorblockA{\IEEEauthorrefmark{1}  College of Information Science and Electronic Engineering, Zhejiang University, Hangzhou, China }
\IEEEauthorblockA{\IEEEauthorrefmark{2} School of Communication Engineering, Hangzhou Dianzi University, Hangzhou, China\\
Email: liangxs@hdu.edu.cn}
\IEEEauthorblockA{\IEEEauthorrefmark{3}Department of Electronic Engineering, Shanghai Jiao Tong University,  Minhang, Shanghai, China\\
Email: \{xbgong, yongpeng.wu\}@sjtu.edu.cn}
\IEEEauthorblockA{\IEEEauthorrefmark{4} School of Electrical Engineering and Telecommunications,
the University of New South Wales, Australia \\Email:w.k.ng@unsw.edu.au}
\IEEEauthorblockA{\IEEEauthorrefmark{5} School of Telecommunication and Information Engineering, Nanjing University of Posts and Telecommunications,\\ Nanjing, China \\Email:hongt@njupt.edu.cn}
}


\maketitle

\begin{abstract}
{N}on-orthogonal multiple access (NOMA) is a promising spectrally-efficient technology to meet the massive  data requirement of the next-generation wireless communication networks.
In this paper, we consider a cooperative non-orthogonal multiple access  (CNOMA) networks consisting of a base station 
 and two users, where the near user serves as a decode-and-forward  relay to help the far user,
 and investigate the outage probability of the CNOMA users  under two different types of channel estimation errors.
For both CNOMA users, we derive the closed-form expressions of the outage probability  and discuss the asymptotic characteristics for the outage probability in the high signal-to-noise ratio (SNR) regimes.
Our results show that for the case of constant variance of the channel estimation error, the outage probability of two users are limited by a performance bottleneck which related to the value of the error variance.
In contrast, there is no such performance bottleneck for the outage probability when the variance of the channel estimation error deceases with SNR,  and in this case the diversity gain is fully achieved by the far user.
\end{abstract}


%
\IEEEpeerreviewmaketitle

\section{Introduction}
%

{N}on-orthogonal multiple access (NOMA) has been considered as an emerging technology which can  address the massive data requirement due to increasing demand of mobile Internet and the Internet-of-Things (IoT) for  the fifth-generation (5G) wireless communications\cite{MShirvanimoghaddam_Massive,Book_KeyTechnologies17}.
Different from the conventional orthogonal multiple access (OMA) schemes, NOMA can  accommodate {substantially} more users  via non-orthogonal
resource allocation to obtain a significant gain in spectral efficiency \cite{YSaito_Nonorthogonal,LDai_Nonorthogonal,SMRIslam_PowerDomain}. 
As an essential technology of 5G to achieve higher spectral efficiency and massive connectivity, NOMA was extended to cooperative transmission to enhance the transmission reliability for the users with poor channel conditions \cite{ZDing_Cooperative,JBKim_Nonorthogonal}.
 Specifically, two types of cooperative NOMA (CNOMA) systems were introduced and classified by different cooperation schemes:
 the cooperation among the NOMA users  \cite{ZDing_Cooperative,YZhou_Dynamic18}, and the  CNOMA systems employing dedicated relays \cite{JBKim_Nonorthogonal,XLiang_Outage17}.
 It is shown that the above two types of CNOMA can utilize the resources of the network efficiently and improve the spectral efficiency of the system compared to cooperative OMA systems\cite{ZDing_ASurveyonNonOrthogonal}.

As perfect channel state information (CSI) cannot be acquired in practice due to the limited overhead of pilot signals in time-division-duplex 
systems and the finite capacity of feedback channel in frequency-division-duplex 
systems, the impact of imperfect CSI on NOMA networks has drawn much attention in recent years.
Some early work addressed NOMA networks based on statistical characteristics of channels.  For instance, the outage performance of a NOMA system was discussed in \cite{ZDing_OnPerformance14} with the priori knowledge of the distribution characteristics of the users's location and their small scale fading.
In \cite{QSun_OnErgodic15}, the ergodic capacity maximization problem was studied for the multiple-input multiple-output (MIMO) NOMA systems with
statistical CSI acquired  by the transmitters.
Afterwards, the impact of channel estimation error was further discussed for practical systems in the following works \cite{ZYang_Performance16,WCai_User16,HVCheng_NOMA17,YGao_Analysis18}.
For instance, the authors in \cite{ZYang_Performance16} investigated the outage probability and the average rate for NOMA systems under channel estimation error, with comparison of the NOMA systems when only statistical CSI is known by the receivers.
In \cite{WCai_User16}, a tractable analysis on the outage probability was performed for a downlink NOMA system with imperfect CSI on account of estimation error and noise. Based on that, the  user selection and power allocation were optimized.
Moreover, the study of NOMA systems with imperfect CSI was extended to multi-users MIMO-NOMA systems in \cite{HVCheng_NOMA17},  in which the rate gain of system was indicated by applying NOMA scheme in MIMO systems.
Besides, the authors in \cite{YGao_Analysis18} proposed a dynamic-ordered self-interference cancellation (SIC) receiver  for NOMA systems with the assumption of  imperfect CSI, then the advantage of the proposed receiver was shown with comparison of traditional SIC receiver under the same situation of imperfect CSI.

Despite the aforementioned progress on the study of NOMA systems under imperfect CSI, the impact of imperfect CSI on CNOMA networks is rarely addressed.
To the best of the authors' knowledge, most of the previous literatures 
on  CNOMA networks assume that perfect CSI can be obtained at the receivers \cite{ZDing_Cooperative,JBKim_Nonorthogonal,YZhou_Dynamic18,XLiang_Outage17}, 
which is infeasible in practical systems. 
In this paper, we investigate the outage performance of a downlink CNOMA system under imperfect CSI, in which two different types of CSI estimation errors, constant and variable variance of CSI estimation errors are assumed. Our contribution includes two parts:\\
  1) For the considered two-user CNOMA system under imperfect CSI, we firstly derive the exact closed-form expression of the outage probability for the near user.
  However,  the closed-form expression of outage probability of the far user is too complicated and mathematical intractable. Therefore we propose an analytical approximation for that, 
  which is validated to be sufficiently close to the exact outage probability of the far user;\\
 2) With the above analytical results for outage probability of both users, we analyze the asymptotic behaviors of the outage probability in the high SNR regimes, and compared them between two types of channel estimation errors with constant and variable variance.
We found that 
 a performance floor exist for the outage probability of both users in the high SNR regimes when the error variance keeps constant.
 On the  contrary, the outage probability of both users will decrease with the SNR when the errors variance is reduced by the increase of SNR. Moreover, we prove that the diversity gain of the far user can be fully achieved in the case of the variable error variance.
Monte-Carlo simulations are provided to validate our  analytical results.

\emph{Notations}---
In this paper,  we denote the probability and the expectation  value of a random event $A$ by ${\tt P}\{A\}$ and ${\tt E}\{A\}$, respectively.
 $| \cdot |$ denotes the absolute value of a complex-valued scalar.

\begin{figure}[!t]
\setlength{\abovecaptionskip}{0.cm}
\setlength{\belowcaptionskip}{-0.cm}
\centering
\includegraphics[scale=0.70]{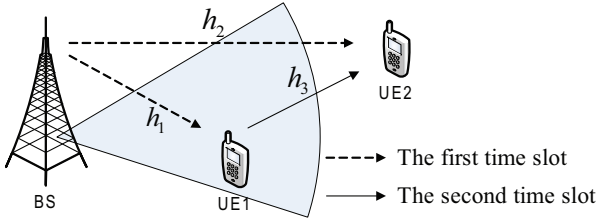}
\caption{A CNOMA system where the near user (UE1) assists the far user (UE2).}
\label{fig:scenario}
\end{figure}

\section{System Model}\label{System_Model}
Consider a model of a downlink CNOMA system as shown  in Fig.~\ref{fig:scenario},
including one BS and  two users (UE1 and UE2), in which UE1 and UE2 directly communicate with  the BS.
 However, UE2 is much far away from the BS than UE1 so that UE2 needs the help from UE1. Each node is equipped with a single-antenna and UE1 serves as a relay (for UE2) operating in HD mode. For this system, we model all the channels as independent Rayleigh fading, and represent the channels between the BS and the two users and the channel between  the two users as  $h_1  \sim \mC\mN\left( {0,\sigma _1^2 } \right)$, $h_2  \sim \mC\mN\left( {0,\sigma _2^2 } \right)$, and $h_3 \sim \mC\mN\left( {0,\sigma _3^2 } \right)$, respectively.
 Besides, we assume that the CSI are estimated imperfectly at the receivers, thus we have \cite{ZYang_Performance16}
\begin{equation}\label{imperfect_CSI}
h_i  = \hat h_i  + e_i,
\end{equation}
where $\hat h_i$ denotes the estimated channel coefficient for $h_i$ with $\hat h_i  \sim \mC\mN\left( {0,\hat \sigma _i^2 } \right)$ and $e_i$ denotes the channel estimation error with $e_i  \sim \mC\mN\left( {0,\sigma _{e_i }^2 } \right)$,  for $i=1,2,3$. In this paper,  two types of the channel estimation error are  considered. The first one is that $\sigma _{e_i }^2 $  decreases inversely proportional to the received SNR by a scaling factor $\eta$ \cite{SSIkki_TwoWay12}, and the second is that $\sigma _{e_i }^2 $  keeps constant \cite{ZYang_Performance16}.\\
  The transmission scheme of CNOMA consists of two consecutive time slots with equal length, as described in the follows.
During the first time slot, the BS broadcasts a superimposed signal, $x_1   = \sqrt {P_1 } s_1  + \sqrt {P_2 } s_2$,
to both users, where $s_1$ and $s_2$ are the desired signals for UE1 and UE2, respectively, with
$\tE\left\{ {\left| {s_1 } \right|^2 } \right\} = \tE\left\{ {\left| {s_2 } \right|^2 } \right\} = 1$.
We denote the  transmit power for UE1 and UE2 as $P_1$ and $P_2$, respectively, and denote the total transmit power for the BS as $ P_{\rm T} = P_1  + P_2 $.
 Therefore the received signals  at the  two users  are expressed by
\begin{align}\label{yi_slot1}
y_i  = \left( {\hat h_i  + e_i } \right)x^{BS}  + n_i,\quad\mbox{for~} i = 1,2
\end{align}
where $n_i $  denotes the complex additive white Gaussian noise (AWGN)  at the UE$i$, i.e., $n_i  \sim \mC\mN\left( {0,N_0 } \right)$.
We assume $\sigma _1^2>\sigma _2^2$ without loss of generality and set $P_1 < P_2$ according to the NOMA protocol described in \cite{ABenjebbour_Concept}.\\
During the second time slot,
 UE1 serves as a DF relay by transmitting $x_2   = \sqrt {P_3 } s_2$ to UE2 after $s_2$ was successfully detected.
This means that there are two cases of the received signal at UE2 as follows.\\
Case 1: When $s_2$ was not detected by UE1 successfully, then $x_2$ cannot be sent by UE1 and the received signal at UE2 is $y_2$ as represented by (\ref{yi_slot1}).\\
Case 2: When $s_2$ was  detected by UE1 successfully, then UE1 transmit $x_2$ to UE2, and the received signal at UE2 is expressed as
\begin{equation}\label{y_3DF}
y_3^{\rm UE2}  =\left( {\hat h_3  + e_3 } \right)x_2   + n_3,
\end{equation}
where $n_3 $ denotes the AWGN  at UE2, i.e., $n_3  \sim \mC\mN\left( {0,N_0 } \right)$. 
Then, by using the maximum ratio combining (MRC), the  received signals from both time slots at UE2 are combined with the MRC coefficients, $\omega _1$ and $\omega _2$,  which yields
\begin{equation}\label{y_3cDF}
y_c^{\rm UE2}  = \omega _1   y_2  + \omega _2  y_{3}^{\rm UE2}.
\end{equation}

Based on the considered model, we can express the received SINRs at UE1  as
\begin{align}
&\gamma _{21} = \frac{{\left| {\hat h_1 } \right|^2 P_2 }}{{\left| {\hat h_1 } \right|^2 P_1  + \left| {e_1 } \right|^2 \left( {P_1  + P_2 } \right) + N_0 }}, \label{Gm_12}\\
&\gamma _1 = \frac{{\left| {\hat h_1 } \right|^2 P_1 }}{{\left| {e_1 } \right|^2 \left( {P_1  + P_2 } \right) + N_0 }} \label{Gm_11},
\end{align}
where $\gamma _{21} $ and $\gamma _1 $, respectively, denote  the SINRs for detecting  $s_2$ and $s_1$ in the SIC process at UE1.
Then, according to the two cases for the received signals at UE2, the SINR at UE2 (for decoding  $s_2$) are shown as follows.\\
Case 1: 
The received signal at UE2 is $y_2$ in (\ref{yi_slot1})
 and the SINR at UE2 is represented by
 \begin{equation}\label{Gm_DF1}
\gamma _2^{\left( 1 \right)}  = \frac{{\left| {\hat h_2 } \right|^2 P_2 }}{{\left| {\hat h_2 } \right|^2 P_1  + \left| {e_2 } \right|^2 \left( {P_1  + P_2 } \right) + N_0 }}.
 \end{equation}
\\Case 2: 
The received signal at UE2 is $y_c^{\rm UE2}$ in (\ref{y_3cDF}), and by using the  property of MRC,
the SINR at UE2  is given by
\begin{small}
 \begin{align}\label{Gm_DF2}
&\gamma _2^{\left( 2 \right)}
=\frac{{\left| {\hat h_2 } \right|^2 P_2 }}{{\left| {\hat h_2 } \right|^2 P_1  + \left| {e_2 } \right|^2 \left( {P_1  + P_2 } \right) + N_0 }} + \frac{{\left| {\hat h_3 } \right|^2 P_3 }}{{\left| {e_3 } \right|^2 P_3  + N_0 }}.
 \end{align}
\end{small}

\section{Outage Probabilities under imperfect CSI}\label{Performance_Analysis}
In this section, we investigate the outage probability, which is an important performance metric of the considered CNOMA system. The outage probability of the CNOMA users are both derived in closed-form expressions.
 Furthermore, the asymptotic characteristics of outage performance of the CNOMA users are investigated for the high SNR regimes and the impact of different CSI estimation errors on the outage  performance of the CNOMA  users are discussed.
\subsection{Outage Probabilities for CNOMA Users}\label{Outage_Probability}
To begin with, we characterize the outage probability achieved by this two-phase CNOMA system. Denoting the SINR thresholds of data rate requirements of UE1 and UE2 as $\bar \gamma_1$ and $\bar \gamma_2$, respectively, the definition of the outage probability of UE1 is expressed as
\begin{equation}\label{Def:PO_UE1}
\tP_{\rm out}^{\rm UE1}  = \tP\left\{ {\gamma _{21} < \bar \gamma _2 {\mbox{ or }}\gamma _1 < \bar \gamma _1 } \right\},
\end{equation}
and that of UE2 is expressed as $\tP_{\rm out}^{\rm UE2}  = \tP\left\{ {\gamma _2^{\rm UE2}  < \bar \gamma _2 } \right\}$.\\
We denote $\lambda_{\rm P}= {\frac{{P_2 }}{{P_1 }}} $ and note that $\lambda_{\rm P} > \bar \gamma _2 $ is a sufficient  condition for this considered CNOMA system working normally (It can be easily proved that $\tP_{\rm out}^{\rm UE1}=\tP_{\rm out}^{\rm UE2}=1$ when $\lambda_{\rm P} \le \bar \gamma _2 $). Therefore we assume $\lambda _{\rm P}> \bar \gamma _2 $ for the analysis of the outage probability in the rest of this paper.
Then, the  outage probability of UE1 is shown in Theorem~\ref{Thm:PO_UE1}.
\begin{theorem}\label{Thm:PO_UE1}
  The  outage probability of UE1  is given by
  \begin{equation}\label{Eq0:PO_UE1}
\tP_{\rm out}^{\rm UE1} =  {1 - \left( {1 + \chi _M\left( {1 + \lambda _{\rm P} } \right) \frac{{\sigma _{e_1 }^2 }}{{\hat \sigma _1^2 }}} \right)^{ - 1} \exp \left( { - \frac{{\chi _M }}{{\rho _{11} }}} \right)},
\end{equation}
where $\rho _{11}  = \frac{{P_1 \hat \sigma _1^2 }}{{N_0 }}$, $\chi  = \frac{{\bar \gamma _2 }}{{\lambda _{\rm P} - \bar \gamma _2 }}$, $\lambda _{\rm P} = \frac{{P_2 }}{{P_1 }}$ and $\chi _M  = \max \left\{ {\chi ,\bar \gamma _1 } \right\}$.
\end{theorem}
\begin{proof}
See Appendix~\ref{Appendix:PO_UE1}.
\end{proof}

Besides,  the overall  outage probability of UE2 is decided by the actual results of the outage probability of UE2 in the two cases, which are shown  in the following lemma.
\begin{lemma}\label{Thm:PO_UE2_DF}

The  outage probability of UE2 in Case~1 
is given by
\begin{align}\label{Eqn1:Pout_UE2_1}
&\tP_{\rm out1}^{\rm UE2}= 1 - \left( {1 + \chi \left( {1 + \lambda _{\rm P} } \right)\frac{{\sigma _{e_2 }^2 }}{{\hat \sigma _2^2 }}} \right)^{ - 1} \exp \left( { - \frac{\chi }{{\rho _{12} }}} \right)
\end{align}
with $\rho _{12}  = \frac{{P_1 \hat \sigma _2^2 }}{{N_0 }}$,
and the  outage probability of UE2 in Case~2
is given by (\ref{Eqn1:Pout_UE2_2}) (at top of next page).

\begin{small}
\begin{figure*}
\setlength{\abovecaptionskip}{0.cm}
\setlength{\belowcaptionskip}{-0.cm}
\begin{align}\label{Eqn1:Pout_UE2_2}
&\tP_{\rm out2}^{\rm UE2}= 1 - \left( {1 + \frac{{\sigma _{e_3 }^2 }}{{\hat \sigma _3^2 }}\bar \gamma _2 } \right)^{ - 1} \exp \left( { - \frac{{\bar \gamma _2 }}{{\rho _3 }}} \right) \nonumber \\
 &-\int_0^{\bar \gamma _2 } {\left[ {\frac{1}{{\rho _3 }}\left( {1 + \frac{{\sigma _{e_3 }^2 }}{{\hat \sigma _3^2 }}y} \right)^{ - 1}  + \frac{{\sigma _{e_3 }^2 }}{{\hat \sigma _3^2 }}\left( {1 + \frac{{\sigma _{e_3 }^2 }}{{\hat \sigma _3^2 }}y} \right)^{ - 2} } \right]\left( {1 + \frac{{\left( {\bar \gamma _2  - y} \right)\left( {1 + \lambda _{\rm P} } \right)}}{{\lambda _{\rm P} - \bar \gamma _2  + y}}\frac{{\sigma _{e_2 }^2 }}{{\hat \sigma _2^2 }}} \right)^{ - 1} } \left[ {\exp \left( { - \frac{{\left( {\bar \gamma _2  - y} \right)}}{{\lambda _{\rm P} - \bar \gamma _2  + y}}\frac{1}{{\rho _{12} }} - \frac{y}{{\rho _3 }}} \right)} \right]dy
   \end{align}
 \hrulefill
\end{figure*}
\end{small}
\end{lemma}
\begin{proof}
See Appendix~\ref{Appendix:PO_UE2_DF}.
\end{proof}
However, the results in (\ref{Eqn1:Pout_UE2_2}) is too complicated and is difficult to be analyzed. Therefore by using  the mean value of $\left| {e_2 } \right|^2$, ${\hat \sigma _2^2 }$, to replace with $\left| {e_2 } \right|^2$ in (\ref{Eqn1:Pout_UE2_2}), we obtain an approximation of (\ref{Eqn1:Pout_UE2_2}) as shown in the following lemma.
\begin{lemma}\label{Thm:PO_UE2_Apprx}
An approximation of the  outage probability of UE2 in Case~2 
is given  by
\begin{small}
 \begin{align}\label{Eqn:PO_UE2_Apprx}
\tilde \tP_{\rm out2}^{\rm UE2} 
  = 1 - \exp \left( { - \frac{{I_{\tilde X} \chi }}{{P_1 \hat \sigma _2^2 }}} \right) - \frac{{I_{\tilde X} }}{{P_1 \hat \sigma _2^2 }}\exp \left( {\frac{{I_{\tilde Y} \left( {\lambda _{\rm P} - \bar \gamma _2 } \right)}}{{P_3 \hat \sigma _3^2 }}} \right)\Theta \left( \chi  \right),
\end{align}
\end{small}
where $ \Theta \left( \chi  \right) = \int_0^\chi  {\exp \left( { - \frac{{I_{\tilde X} }}{{P_1 \hat \sigma _2^2 }}u - \frac{{I_{\tilde Y} \lambda _{\rm P} }}{{P_3 \hat \sigma _3^2 }}\frac{1}{{1 + u}}} \right)du}$ with $\chi  = \frac{{\bar \gamma _2 }}{{\lambda _{\rm P} - \bar \gamma _2 }}$, $\lambda _{\rm P} = \frac{{P_2 }}{{P_1 }} $,
 $I_{\tilde X}  = \sigma _{e_2 }^2 \left( {P_1  + P_2 } \right) + N_0 $, $I_{\tilde Y}  = \sigma _{e_3 }^2 P_3  + N_0$.
\end{lemma}
\begin{proof}
See Appendix~\ref{Appendix:PO_UE2_Apprx}.
\end{proof}

With the conclusions of Lemmas~\ref{Thm:PO_UE2_DF}--\ref{Thm:PO_UE2_Apprx}, we finally obtain the overall  outage probability of UE2 as follows.
\begin{theorem}\label{Thm:PO_UE2_Ovr}
The overall  outage probability of UE2 is approximated by (\ref{Eqn0:PO_UE2_Ovr}) (at top of next page)
\begin{figure*}
\begin{equation}\label{Eqn0:PO_UE2_Ovr}
 \tilde \tP_{\rm ovr}^{\rm UE2}  = 1 - \mu _2 \exp \left( { - \frac{\chi }{{\rho _{12} }}} \right) + \mu _1 \mu _2 \exp \left( { - \frac{\chi }{{\rho _{11} }} - \frac{\chi }{{\rho _{12} }}} \right) - \mu _1 \exp \left( { - \frac{\chi }{{\rho _{11} }} - \frac{\chi }{{\rho _{\tilde X} }}} \right) - \frac{{\mu _1 }}{{\rho _{\tilde X} }}\exp \left( {\frac{{\lambda _{\rm P} - \bar \gamma _2 }}{{\rho _{\tilde Y} }} - \frac{\chi }{{\rho _{11} }}} \right)\Theta \left( \chi  \right)
\end{equation}
 \hrulefill
\end{figure*}
with $\rho _{\tilde X}  = \frac{{P_1 \hat \sigma _2^2 }}{{I_{\tilde X} }}$, $\rho _{\tilde Y}  = \frac{{P_3 \hat \sigma _3^2 }}{{I_{\tilde Y} }}$, $\mu _i  = \left( {1 + \chi \left( {1 + \lambda _{\rm P} } \right)\frac{{\sigma _{e_i }^2 }}{{\hat \sigma _i^2 }}} \right)^{ - 1} $ and $\rho _{1i}  = \frac{{P_1 \hat \sigma _i^2 }}{{N_0 }}$, $i=1,2$.
\end{theorem}

\begin{proof}
See Appendix~\ref{Appendix:PO_UE2_Ovr}.
\end{proof}

\subsection{Asymptotic characteristics for CNOMA users}\label{PA_UE2}
 Since we have obtained the analytical expressions of outage probability of the two users, we then carry on a further study on  asymptotic characteristics of $\tP_{\rm out}^{\rm UE1}$ and $\tilde \tP_{\rm out2}^{\rm UE2}$ in the high SNR regimes. In specific,  assuming that $\rho _{11}$ and $\rho _{12}$ both grow to infinity while $\rho _3  = c_3 \rho _{12}$ with a constant ratio of $c_3$, the asymptotic characteristics of the outage probability of both users are discussed under two different CSI estimation error assumptions. The analytical results are shown in the following theorem.

\begin{theorem}\label{Thm:Asymp_Characs}
When $\sigma _{e_i }^2$ is a constant, i.e., $ {\sigma _{e_i }^2  = \sigma _c^2 }$, $\forall i $, the outage probability of the two users both approach performance floors in the high SNR regimes, which are given by
\begin{align}
&\mathop {\lim }\limits_{\scriptstyle \rho _{11}  \to \infty  \hfill}
\tP_{\rm out}^{\rm UE1}\approx
 \frac{{\left( {1 + \lambda _P } \right)\chi _M \frac{{\sigma _c^2 }}{{\hat \sigma _1^2 }}}}{{1 + \left( {1 + \lambda _P } \right)\chi _M \frac{{\sigma _c^2 }}{{\hat \sigma _1^2 }}}}\label{Asymp_Pout_UE1fix},\\
 &\mathop {\lim }\limits_{\scriptstyle \rho _{12}  \to \infty  \hfill}
\tilde\tP_{\rm out}^{\rm UE2}\approx
1 - \mu _2^c  + \mu _1^c \mu _2^c  - \mu _1^c \exp \left( {\tilde \varepsilon _c } \right)\label{Asymp_Pout_UE2fix}
\end{align}
\begin{figure*}
\begin{equation}\label{Eqn0:Esp_Asym}
\tilde \varepsilon _c  = \ln \left[ {\exp \left( { - \frac{{\sigma _c^2 }}{{\hat \sigma _2^2 }}\left( {1 + \lambda _P } \right)\chi } \right) + \frac{{\sigma _c^2 }}{{\hat \sigma _2^2 }}\left( {1 + \lambda _P } \right)\exp \left( { - \frac{{\sigma _c^2 }}{{\hat \sigma _3^2 }}\frac{{\lambda _P }}{{1 + \chi }}} \right)\int_0^\chi  {\exp \left( { - \frac{{\sigma _c^2 }}{{\hat \sigma _2^2 }}\left( {1 + \lambda _P } \right)u - \frac{{\sigma _c^2 }}{{\hat \sigma _3^2 }}\frac{{\lambda _P }}{{1 + u}}} \right)du} } \right]
\end{equation}
 \hrulefill
\end{figure*}
with  $\tilde \varepsilon _c $ in (\ref{Eqn0:Esp_Asym}) (at top of next page) and $\mu _i^c  = \left( {1 + \chi \left( {1 + \lambda _P } \right)\frac{{\sigma _c^2 }}{{\hat \sigma _i^2 }}} \right)^{ - 1}$, $i = 1,2$.\\
Besides, when $\sigma _{e_i }^2$ is inversely proportional to the received SNR, i.e., ${\sigma _{e_i }^2  = \eta \frac{{N_0 }}{{P_i \sigma _i^2 }}}$, $\forall i $, the outage probability of the two users are given by
\begin{align}
&\mathop {\lim }\limits_{\scriptstyle \rho _{11}  \to \infty  \hfill}
\tP_{\rm out}^{\rm UE1}\approx
{\frac{{\chi _M }}{{\rho _{11} }}\left( {\left( {1 + \lambda _P } \right)\frac{\eta }{{\hat \sigma _1^2 }} + 1} \right)}
\label{Asymp_Pout_UE1var},\\
 &\mathop {\lim }\limits_{\scriptstyle \rho _{12}  \to \infty  \hfill}
\tilde\tP_{\rm out}^{\rm UE2}\approx
\rho_{12}^{-2}\left( {\frac{{\eta \left( {1 + \lambda _P } \right)}}{{\lambda _P \hat \sigma _2^2 }} + 1} \right)\nonumber \\&
  \cdot \left[ {\left( {\frac{{\eta \left( {1 + \lambda _P } \right)}}{{\lambda _P \hat \sigma _2^2 }} + 1} \right)\frac{{\chi ^2 }}{2} - \frac{{\lambda _P }}{{c_3 }}\left( {\frac{\eta }{{\hat \sigma _3^2 }} + 1} \right)\ln \left( {1 + \chi } \right)} \right]\label{Asymp_Pout_UE2var}.
\end{align}

It can be easily proved  from (\ref{Asymp_Pout_UE1fix}) and (\ref{Asymp_Pout_UE2fix}) that there are performance floors for the outage probability of both CNOMA users in the high SNR regimes when the CSI estimation error variance keeps constant. Meanwhile, we can observe that the performance floors are independent of the received SNRs, and the floors increase with the error variance.
However, there is no such performance bottleneck in the case of variable  variance of CSI estimation error, and it is indicated from (\ref{Asymp_Pout_UE1var}) and (\ref{Asymp_Pout_UE2var})  that the outage probability of the two users both decrease with the increase of SNR. Moreover, it can be proved   from (\ref{Asymp_Pout_UE2var}) that the diversity order for UE2 is $2$, which means the diversity gain is fully acquired by UE2 under variable error variance.
\end{theorem}
\begin{proof}
Depending on the different types of CSI estimation errors, we substitute $ {\sigma _{e_1 }^2  = \sigma _c^2 }$ or ${\sigma _{e_1 }^2  = \eta \frac{{N_0 }}{{P_1 \sigma _1^2}}}$, respectively into (\ref{Eq0:PO_UE1}) and omit the  the small terms in them for $\rho _{{{11}}}  \to \infty $. Then by skipping the details of simplification, we obtain (\ref{Asymp_Pout_UE1fix}) and (\ref{Asymp_Pout_UE1var}) for UE1.
In the sequel, we substitute $ {\sigma _{e_i }^2  = \sigma _c^2 }$ or ${\sigma _{e_i }^2  = \eta \frac{{N_0 }}{{P_i \sigma _i^2}}}$, respectively into (\ref{Eqn0:PO_UE2_Ovr}), for $i=2,3$. Then we apply the Taylor's expansions on the logarithmic and exponential functions in them and omit the small terms for $\rho _{{{12}}}  \to \infty $. Finally, we obtain (\ref{Asymp_Pout_UE2fix}) and (\ref{Asymp_Pout_UE2var}) for UE2.
\end{proof}

%

\section{Numerical Results}\label{Numerical_Results}

\begin{figure}[htbp]
\setlength{\abovecaptionskip}{0.cm}
\setlength{\belowcaptionskip}{-0.cm}
\centering
\includegraphics[scale=0.60]{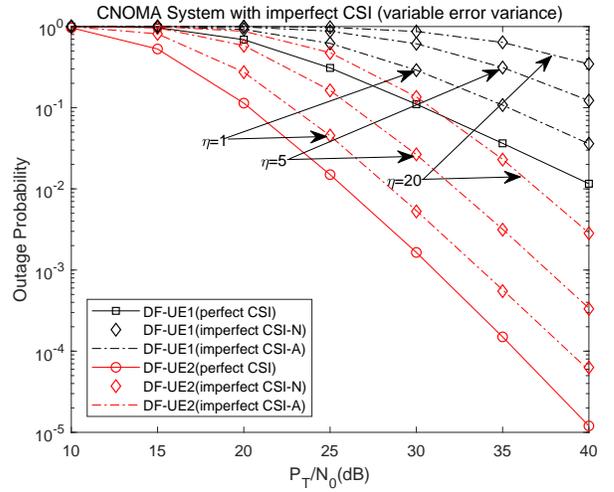}
\caption{The outage probabilities for two users with variable variance of CSI estimation error.}
\label{fig:result_1}
\end{figure}

\begin{figure}[htbp]
\setlength{\abovecaptionskip}{0.cm}
\setlength{\belowcaptionskip}{-0.cm}
\centering
\includegraphics[scale=0.60]{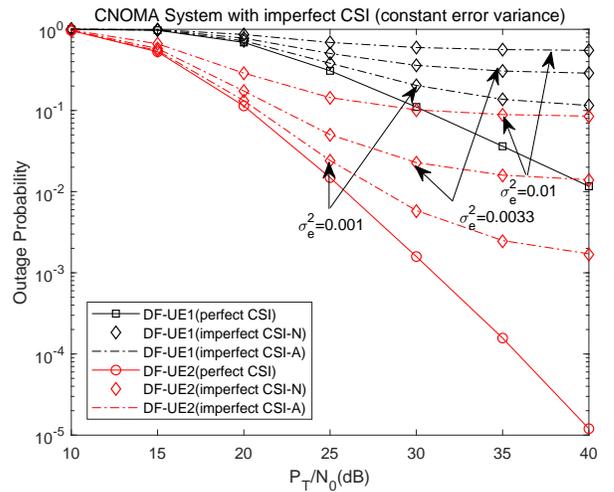}
\caption{The outage probabilities for two users with constant variance of CSI estimation error.}
\label{fig:result_2}
\end{figure}

In this section, the outage probability of the CNOMA system is evaluated based on Monte-Carlo simulations averaging over $10^6$ independent channel realizations, considering different levels of channsl estimation error variances.
Furthermore, the asymptotic characteristics of the outage probability of the users are shown in the high SNRs.
The system performance with perfect CSI is also included as a benchmark.
In specific, the adopted simulation system parameters are set as follows.
The data requirements of UE1 and UE2 are set as $R_1=1.5$~bit/Hz/s and $R_2=1$~bit/Hz/s, respectively. The variances of channels fading for BS-UE1, BS-UE2, and UE1-UE2 are set as $\sigma^2_1=0.36$, $\sigma^2_2=0.16$, and $\sigma^2_3=0.64$, respectively. In addition, the power of noises are all set as $N_0=1$,
 and the parameters of the transmit powers are set by $\lambda_P =P_2/P_1=5$ and $P_3 = P_{\rm T}-5$~dB.
  In all the figures, the numerical results of outage probability are labeled '-N', and the analytical results of outage probability of UE1 (in (\ref{Eq0:PO_UE1})) and UE2 (in (\ref{Eqn0:PO_UE2_Ovr})) are labeled '-A'. We set $\sigma _{e_i }^2  = \eta N_0 \left( {P_i \sigma _i^2 } \right)^{ - 1}$  with $\eta  = 1,5$ and $20$ for the case of variable CSI estimation error variance, and set $\sigma _{e_i }^2  = 1 \times 10^{-3},3.3\times 10^{-3}$ and $ 1\times 10^{-2}$ for the case of constant CSI estimation error variance, for $i = 1,2,3$.

In Fig.~\ref{fig:result_1},  the outage probability of the two users under variable CSI estimation error variance are shown as a function of $\eta$. 
An excellent agreement between the analytical and the Monte Carlo simulation results  of outage probability of both users, which verifies the validity of the approximation for the  outage probability of UE2.
Compared with the CNOMA systems with perfect CSI, there is an increase of the outage probability of both users with imperfect CSI and the gap between them increases with $\eta$.
In addition, it is observed that the curves of the same user show the same slope degree (in log-scale) in the moderate-to-high SNR range, which indicates the outage performance with variable CSI estimation error variance can achieve the same diversity order with that of perfect CSI estimation, as been proved in Theorem~\ref{Thm:Asymp_Characs}.

 {Fig.~\ref{fig:result_2} shows the variation of the outage probability of two CNOMA users under constant CSI estimation error variance, following the variation of 
 $\sigma _{e_i }^2 $.
The analytical results of the outage probability of both users in Fig~\ref{fig:result_2} show an excellent agreement with the Monte-Carlo simulation results of both users.
It is also observed from Fig.~\ref{fig:result_2} that there is an increase of the outage probability of both users with imperfect CSI in comparison with that with perfect CSI, while the gap between them increases with $\sigma _{e_i }^2 $.
Moreover, it is indicated from Fig.~\ref{fig:result_2} that there is a  floor for outage probability of both users under constant variance of CSI estimation error in high SNR regimes. This means that the constant variance of CSI estimation error causes a bottleneck of the outage performance in CNOMA systems, which validates the conclusion of Theorem~\ref{Thm:Asymp_Characs}.
}

\section{Conclusion}

In this paper,  we have studied the impact of imperfect CSI on the outage probability of downlink two-user CNOMA networks in which the near user acts as a DF relay for assisting the far user. For two different types of imperfect CSI estimation when the variance of the estimation error keeps constant value or decreases linearly with received SNR, the closed-form expressions of the outage probability of the two users were  derived.  Based on that, the asymptotic behaviors  of the outage probability of both users were investigated.
It is shown that the outage probability of both users under constant variance of errors approaches a performance bottleneck, even when the received SNRs are sufficiently high.
Meanwhile, there is no such performance bottleneck for the outage probability  of the two users when the error variance improves with the received SNRs. It is also shown that the diversity order of both users for the case of variable error variance are identical with that under perfect CSI.

\appendices
\section{Proof of Theorem\ref{Thm:PO_UE1}}\label{Appendix:PO_UE1}
Firstly, we note that $\tP_{\rm out}^{\rm UE1} $ in (\ref{Def:PO_UE1}) can be rewritten as (\ref{EqnA1:PO_UE1}) (at top of next page)
\begin{figure*}
\begin{small}
\begin{equation}\label{EqnA1:PO_UE1}
 \tP_{\rm out}^{\rm UE1}  = \tP\left\{ {\frac{{\left| {\hat h_1 } \right|^2 P_2 }}{{\left| {\hat h_1 } \right|^2 P_1  + \left| {e_1 } \right|^2 \left( {P_1  + P_2 } \right) + N_0 }} < \bar \gamma _2 {\mbox{ or }}\frac{{\left| {\hat h_1 } \right|^2 P_1 }}{{\left| {e_1 } \right|^2 \left( {P_1  + P_2 } \right) + N_0 }} < \bar \gamma _1 } \right\},
\end{equation}
\end{small}
 \hrulefill
\end{figure*}
and further recast into
\begin{equation}\label{EqnA2:PO_UE1}
 \tP_{\rm out}^{\rm UE1}  =  \tP\left\{ {\left| {\hat h_1 } \right|^2  < \beta _1 \left( {\left| {e_1 } \right|^2 \left( {P_1  + P_2 } \right) + N_0 } \right)} \right\}
\end{equation}
with $\beta _1  = \max \left\{ {\frac{{\bar \gamma _2 }}{{P_2  - P_1 \bar \gamma _2 }},\frac{{\bar \gamma _1 }}{{P_1 }}} \right\}$.\\ 
As $\hat h_1  \sim \mC\mN\left( {0,\hat \sigma _1^2 } \right)$, $e_1  \sim \mC\mN\left( {0,\sigma _{e_1 }^2 } \right)$,
we can use the probability distribution function (PDF) of $\left| \hat {h_1 } \right|^2 $
, which is
\begin{align}\label{PDF_hat_h1}
&f_{\left| {\hat h_1 } \right|^2 } \left( x \right) = \frac{1}{{\hat \sigma _1^2 }}\exp \left( { - \frac{x}{{\hat \sigma _1^2 }}} \right),
\end{align}
and the PDF of ${\left| {e_1 } \right|^2 }$, which is
\begin{align}\label{PDF_e1}
&f_{\left| {e_1 } \right|^2 } \left( x \right) = \frac{1}{{\sigma _{e_1 }^2 }}\exp \left( { - \frac{x}{{\sigma _{e_1 }^2 }}} \right)
,
\end{align}
to obtain the expression of (\ref{EqnA2:PO_UE1}).
Skipping the tedious details for derivation,
the final expression of $\tP_{\rm out}^{\rm UE1} $ is given by (\ref{Eq0:PO_UE1})
, which completes the proof.

\section{Proof of Lemma~\ref{Thm:PO_UE2_DF} }\label{Appendix:PO_UE2_DF}

Firstly,
 we consider Case 1: When  $s_2$ was not detected by UE1, the expression of $\tP_{\rm out1}^{\rm UE2} $ is given by
 \begin{equation}\label{EqnA1:PO_UE2_DF}
 \tP_{\rm out1}^{\rm UE2}  = \tP\left\{ {\frac{{\left| {\hat h_2 } \right|^2 P_2 }}{{\left| {\hat h_2 } \right|^2 P_1  + \left| {e_2 } \right|^2 \left( {P_1  + P_2 } \right) + N_0 }} < \bar \gamma _2 } \right\}
 \end{equation}
and further recast into
 \begin{equation}\label{EqnA2:PO_UE2_DF}
 \tP_{\rm out1}^{\rm UE2}  = \tP\left\{ {\left| {\hat h_2 } \right|^2  < \frac{\chi }{{P_1 }}\left( {\left| {e_2 } \right|^2 \left( {P_1  + P_2 } \right) + N_0 } \right)} \right\}.
 \end{equation}
We note that (\ref{EqnA2:PO_UE2_DF}) has a similar form with (\ref{EqnA2:PO_UE1}), and it can be easily to obtain  the expression of $\tP_{\rm out1}^{\rm UE2}$ in (\ref{Eqn1:Pout_UE2_1}).

Next,  we consider Case 2: When  $s_2$ was successfully detected by UE1 and then transmitted to UE2,
 the  outage probability of UE2 is given by
\begin{equation}\label{EqnB2:PO_UE2_DF}
 \tP_{\rm out2}^{\rm UE2}  =  \tP\left\{ {X + Y < \bar \gamma _2 } \right\}=\int_0^{\bar \gamma _2 } {f_Y \left( y \right)} F_X \left( {\bar \gamma _2  - y} \right)dy
 \end{equation}
 with $ X = \frac{{\left| {\hat h_2 } \right|^2 P_2 }}{{\left| {\hat h_2 } \right|^2 P_1  + \left| {e_2 } \right|^2 \left( {P_1  + P_2 } \right) + N_0 }}$, $Y = \frac{{\left| {\hat h_3 } \right|^2 P_3 }}{{\left| {e_3 } \right|^2 P_3  + N_0 }}$.\\
With $\hat h_i  \sim CN\left( {0,\hat \sigma _i^2 } \right)$, $e_i  \sim CN\left( {0,\sigma _{e_i }^2 } \right)$, $i=2,3$,  the cumulative  distribution function (CDF) of $X$ and the PDF of $Y$ are given by (\ref{EqnA:PDF_X}) and (\ref{EqnA:PDF_Y}) (at top of next page), respectively.
\begin{figure*}
\begin{align}
&F_X \left( x \right) = \left\{ {\begin{array}{*{20}c}
   {1 - \left( {1 + \frac{{x\left( {1 + \lambda _{\rm P} } \right)}}{{\lambda _{\rm P} - x}}\frac{{\sigma _{e_2 }^2 }}{{\hat \sigma _2^2 }}} \right)^{ - 1} \exp \left( { - \frac{x}{{\lambda _{\rm P} - x}}\frac{1}{{\rho _{12} }}} \right),\quad 0 \le x < \lambda _{\rm P} }  \\
   {1,\quad x > \lambda _{\rm P} }  \\
\end{array}} \right.
,\label{EqnA:PDF_X}\\
& f_Y \left( y \right) = \left[ {\frac{1}{{\rho _3 }}\left( {1 + \frac{{\sigma _{e_3 }^2 }}{{\hat \sigma _3^2 }}y} \right)^{ - 1}  + \frac{{\sigma _{e_3 }^2 }}{{\hat \sigma _3^2 }}\left( {1 + \frac{{\sigma _{e_3 }^2 }}{{\hat \sigma _3^2 }}y} \right)^{ - 2} } \right]\exp \left( { - \frac{y}{{\rho _3 }}} \right), \quad y \ge 0.\label{EqnA:PDF_Y}
 \end{align}
  \hrulefill
\end{figure*}
 By substituting (\ref{EqnA:PDF_X}) and (\ref{EqnA:PDF_Y}) into (\ref{EqnB2:PO_UE2_DF}), we obtain $ \tP_{\rm out2}^{\rm UE2} $ in (\ref{Eqn1:Pout_UE2_2})
, which completes the proof.

\section{Proof of Lemma~\ref{Thm:PO_UE2_Apprx}}\label{Appendix:PO_UE2_Apprx}
Firstly, the  outage probability of UE2 is approximated by 
\begin{align}\label{EqnB2:PO_UE2_Approx}
 \tilde \tP_{\rm out2}^{\rm UE2}  =  \tP\left\{ {\tilde X + \tilde Y < \bar \gamma _2 } \right\}
=\int_0^{\bar \gamma _2 } {f_{\tilde Y} \left( y \right)} F_{\tilde X} \left( {\bar \gamma _2  - y} \right)dy
 \end{align}
 with $ \tilde X = \frac{{\left| {\hat h_2 } \right|^2 P_2 }}{{\left| {\hat h_2 } \right|^2 P_1  + \sigma _{e_2 }^2 \left( {P_1  + P_2 } \right) + N_0 }}$, $\tilde Y = \frac{{\left| {\hat h_3 } \right|^2 P_3 }}{{\sigma _{e_3 }^2 P_3  + N_0 }}$, 
whilst the CDF of $\tilde X$ and the PDF of $\tilde Y$ are given as follows.

\begin{align}
&
F_{\tilde X} \left( x \right) = \left\{ {\begin{array}{*{20}c}
   {1 - \exp \left( { - \frac{{I_{\tilde X} x}}{{\hat \sigma _2^2 \left( {P_2  - P_1 x} \right)}}} \right),{\rm{0}} \le x < \lambda _{\rm P} }  \\
   {1,x \ge \lambda _{\rm P} }  \\
\end{array}} \right.
,\label{EqnA:CDF_TX}\\
& f_{\tilde Y} \left( y \right) = \frac{{I_{\tilde Y} }}{{P_3 \hat \sigma _3^2 }}\exp \left( { - \frac{{I_{\tilde Y} }}{{P_3 \hat \sigma _3^2 }}y} \right),y \ge 0\label{EqnA:PDF_TY}
 \end{align}
with $I_{\tilde X}  = \sigma _{e_2 }^2 \left( {P_1  + P_2 } \right) + N_0 $ and $I_{\tilde Y}  = \sigma _{e_3 }^2 P_3  + N_0$.\\
 By substituting (\ref{EqnA:CDF_TX}) and (\ref{EqnA:PDF_TY}) into (\ref{EqnB2:PO_UE2_Approx}), finally the expression of $\tilde P_{out}^{UE2}$ is given by (\ref{Eqn:PO_UE2_Apprx})
, which completes the proof.

\section{Proof of Theorem \ref{Thm:PO_UE2_Ovr}}\label{Appendix:PO_UE2_Ovr}
Firstly, the approximation of the overall  outage probability of UE2 is expressed as
\begin{equation}\label{proof:PO_UE2_Ovr}
 \tilde \tP_{\rm ovr}^{\rm UE2}  = \tP_{\rm un}^{\rm UE1}  \cdot \tP_{\rm out1}^{\rm UE2}  + \left( {1 - \tP_{\rm un}^{\rm UE1} } \right)\tilde \tP_{\rm out2}^{\rm UE2}
\end{equation}
where $\tP_{\rm un}^{\rm UE1}$ represents the probability of unsuccessfully decoding $s_2$ at UE1.
Then, $\tP_{\rm un}^{\rm UE1}$ is recast into
\begin{align}\label{Proof:un_Pout_UE1}
\tP_{\rm un}^{\rm UE1} = \tP\left\{ {\left| {\hat h_1 } \right|^2  < \frac{\chi }{{P_1 }}\left( {\left| {e_1 } \right|^2 \left( {P_1  + P_2 } \right) + N_0 } \right)} \right\}.
\end{align}
Due to the similarity between (\ref{EqnA2:PO_UE2_DF}) and (\ref{Proof:un_Pout_UE1}), it can be easily obtained that
\begin{align}\label{Proof1:un_Pout_UE1}
\tP_{\rm un}^{\rm UE1} = 1 - \left( {1 + \chi \left( {1 + \lambda _{\rm P} } \right)\frac{{\sigma _{e_1 }^2 }}{{\hat \sigma _1^2 }}} \right)^{ - 1} \exp \left( { - \frac{\chi }{{\rho _{11} }}} \right).
\end{align}
Then by substituting (\ref{Proof1:un_Pout_UE1}), (\ref{Eqn1:Pout_UE2_2}) and (\ref{Eqn:PO_UE2_Apprx}) into  (\ref{proof:PO_UE2_Ovr}) and through simplifications, we finally obtain $ \tilde \tP_{\rm ovr}^{\rm UE2} $ in
(\ref{Eqn0:PO_UE2_Ovr}), which completes the proof.


\section*{Acknowledgment}

The paper has been accepted by 2018  Information Technology and Mechatronics Engineering Conference (ITOEC 2018£©, 17th Sept. 2018.



%
\bibliographystyle{IEEEtran}
\bibliography{IEEEabrv,reference,MIMO_Relay,NOMA-Relay}

\end{document}